\documentclass[10pt,twocolumn]{IEEEtran}
\usepackage[utf8]{inputenc}
\usepackage{amsmath,amssymb}
\usepackage{graphicx}
\usepackage{cite}
\usepackage{booktabs}
\usepackage{multirow}
\usepackage{algorithm}
\usepackage[dvips]{epsfig}
\usepackage{epsfig}
\usepackage[dvips]{graphicx}
\usepackage{algpseudocode}
\usepackage{hyperref}
\usepackage{textcomp}
\usepackage{xcolor}
\usepackage{bm}
\usepackage{cite}
\usepackage{mathtools}
\usepackage{float}
\usepackage{epstopdf}
\usepackage{algorithm}
\usepackage{algpseudocode}

\newcommand{\ket}[1]{\left|#1\right\rangle}
\newcommand{\bra}[1]{\left\langle#1\right|}
\newcommand{\braket}[2]{\left\langle#1\middle|#2\right\rangle}
\newcommand{\ketbra}[2]{\left|#1\right\rangle\!\left\langle#2\right|}
\newcommand{\Tr}{\operatorname{Tr}}
\newcommand{\eop}{\mathcal{E}}
\newtheorem{theorem}{Theorem}

\begin{document}


\title{\huge Performance Analysis of Multiparty GHZ-State Quantum Conference Key Agreement over Optical Fiber Channel}

\author{Dipanjan~Bera~\IEEEmembership{Student Member,~IEEE}, Ramniwas~Meena, and~Neel~Kanth~Kundu~\IEEEmembership{Member,~IEEE} \\
\thanks{This work was supported in part by the National Quantum Mission of India, INSPIRE Faculty Fellowship (Reg. No.: IFA22-ENG 344), ANRF Prime Minister Early Career Research Grant (ANRF/ECRG/2024/000324/ENS), and IIT Delhi's New Faculty Seed Grant. }

\thanks{D. Bera and R. Meena are with the Centre for Applied Research in Electronics (CARE), Indian Institute of Technology Delhi, New Delhi 110016, India (e-mails: crz258444@care.iitd.ac.in,  ram.21phy@gmail.com.)}

\thanks{N. K. Kundu is with CARE and the Bharti School of Telecommunication Technology and Management, Indian Institute of Technology Delhi, New Delhi-110016, India (e-mail: neelkanth@iitd.ac.in). He is also an honorary fellow at the Department of Electrical and Electronic Engineering, University of Melbourne, Melbourne, VIC, Australia. \\
This work has been submitted to the IEEE for possible publication. Copyright may be transferred without notice, after which this version may no longer be accessible.}
}

\maketitle



\begin{abstract}

Quantum conference key agreement (QCKA) enables multiple spatially
separated users to establish a common key with composable
information-theoretic security by exploiting multipartite quantum
correlations. Greenberger--Horne--Zeilinger (GHZ) states are a
natural resource for this task, but their distribution over optical
fiber is degraded by transmission loss, polarization decoherence, and
detector imperfections. We investigate the finite-key performance of
an $N$-party GHZ-based QCKA system using the $N$-BB84 protocol, where each optical link is modeled by extending the bipartite
erasure--Pauli channel to incorporate detector dark counts. We characterize the impact of this quantum channel on the shared GHZ state and derive closed-form expressions for the error rates. The analytical results reveal that the local bit-error rate depends only on transmissivity and dark counts, while the global phase-error rate
compounds over all $N-1$ links and therefore degrades faster with
both distance and the number of parties.
These error rates are integrated into a
finite-key secret-key-rate (SKR) framework, in which we optimize the
parameter-estimation probability via Brent's derivative-free method.
Numerical results show that increasing fiber
distance and network size reduce the achievable SKR. Furthermore, the
optimized parameter-estimation probability consistently increases the
SKR and extends the maximum secure distance relative to a
fixed-allocation baseline, with comparison to the asymptotic limit
quantifying the residual finite-key penalty.
\end{abstract}

\begin{IEEEkeywords}
Quantum conference key agreement, multipartite quantum key distribution, GHZ state, finite-key analysis, optical fiber, erasure--Pauli channel, polarization mode dispersion, secret key rate optimization.
\end{IEEEkeywords}



\section{Introduction}
\label{sec:introduction}
The security of modern communication systems relies predominantly on computational hardness assumptions, which are known to be vulnerable to attacks by sufficiently large-scale quantum computers. This vulnerability motivates continued interest in cryptographic primitives whose security does not depend on unproven computational assumptions. Quantum key distribution (QKD) addresses this need by allowing two distant parties to establish a secret key whose security follows from the laws of quantum mechanics rather than from computational hardness. Since the original BB84 protocol was proposed, QKD has evolved from a point-to-point cryptographic primitive into a foundational building block of emerging quantum network architectures.

Conventional QKD protocols, including BB84 and its many variants, are inherently bipartite: they establish a shared secret key between exactly two legitimate parties. Many realistic applications, however, involve a group of geographically separated participants who require a single common secret key rather than a collection of independent pairwise keys, representative examples include secure teleconferencing, distributed control of critical infrastructure, and collaborative multi-user quantum networks. Quantum conference key agreement (QCKA), also referred to as multipartite QKD, extends the two-party QKD task to this setting: it allows $N$ mutually trusted users to establish a single shared secret key, exploiting genuinely multipartite quantum correlations, while ensuring that no unauthorized party can obtain useful information about the resulting key~\cite{murta2020quantum}.

Existing approaches to QCKA can be broadly grouped into entanglement-based and prepare-and-measure schemes. Among entanglement-based approaches, the Greenberger--Horne--Zeilinger (GHZ) state has emerged as a natural multipartite resource, since measuring every party's share of a shared GHZ state in the computational basis yields perfectly correlated, uniformly random outcomes across all $N$ parties allowing a single distributed state to directly seed an $N$-party key, rather than $N-1$ separate bipartite keys that would need to be combined. Building on the early observation of Epping~\emph{et al.} that multipartite entanglement can accelerate key distribution in networks with bottlenecks \cite{epping2017}, this line of work has developed in several directions: toward anonymity-preserving protocols, in which parties can establish a conference key without revealing their identities to one another~\cite{hahn2020anonymous,grasselli2022secureanonymous}; toward protocol variants aimed at practical deployment, including repeater-like, measurement-device-independent constructions that relax timing requirements~\cite{lu2025repeaterlike}, and schemes based on cluster states rather than GHZ states~\cite{sun2016cluster}; and toward experimental demonstrations of GHZ-based QCKA over metropolitan-scale fiber and in deployed quantum networks \cite{proietti2021,liu2023,pickston2023}. Within the entanglement-based family, the $N$-BB84 protocol is a direct $N$-party generalization of BB84: a designated sender distributes an $N$-qubit GHZ state to itself and $N-1$ receivers, each party independently selects a measurement basis, the parties publicly reconcile their basis choices and estimate the relevant error rates from a random test sample, and the remaining measurement outcomes are processed via classical error correction and privacy amplification to yield a common secret conference key. We emphasize that $N$-BB84 itself is an established protocol; the present work is concerned with its behavior under a physically realistic channel model rather than with proposing a new protocol.

The performance of any QCKA scheme is fundamentally governed by the quantum channel over which the entangled state is distributed.
Fiber-based transmission introduces several distinct physical impairments, such as photon loss due to fiber attenuation, decoherence due to polarization-mode dispersion (PMD), and spurious detection events due to detector dark counts, each of which affects the shared quantum state differently and consequently degrades the achievable quantum bit error rate (QBER), the inter-party correlations, and the resulting secret-key rate (SKR). This sensitivity to channel imperfections is expected to be more pronounced in the multiparty setting than in bipartite QKD, since a GHZ state must be distributed intact across $N-1$ independent fiber links, and since the residual multipartite coherence of the state generally depends on the combined effect of all $N-1$ links. Therefore,  an increasing number of receivers is expected to impose progressively tighter constraints on the tolerable channel noise and the achievable secure transmission distance.

Despite this, the existing QCKA literature has largely treated channel noise as an abstract, externally specified parameter rather than deriving it from an underlying physical transmission model. The anonymity-preserving protocols of \cite{hahn2020anonymous,grasselli2022secureanonymous} are analyzed under an ideal or generically noisy GHZ state, without characterizing how much noise a real fiber link of a given length would actually introduce. The protocol-oriented works of \cite{lu2025repeaterlike,sun2016cluster} quantify loss through a single transmission efficiency, without connecting it to physical fiber attenuation, polarization decoherence, or detector imperfections. The experimental demonstrations of \cite{proietti2021,liu2023,pickston2023} report the QBER observed in a specific experimental setup, but do not provide a general expression predicting the QBER at other distances or for a different number of parties. Even the finite-key analysis of Li~\emph{et al.}~\cite{li2021asymmetric}, which allows each Bob's channel to differ, treats each error rate as an independently specified parameter: it does not distinguish the \emph{local} bit error between Alice and an individual Bob from the \emph{global} phase error associated with the full GHZ state, nor does it derive either quantity from the physics of the transmission channel. 

A complementary body of work has instead modeled the physical transmission channel directly, but predominantly for the bipartite case: detector dark counts have been modeled as randomizing a fraction of otherwise-lost detection events \cite{shibata2014,korzh2015,seshadreesan2016}, and PMD has been established as the dominant polarization-decoherence mechanism in long single-mode fiber~\cite{poole1986,wai1996,Agrawal2012}. These two effects were unified only in a recent work by Pirandola~\cite{pirandola2026}, which introduced a physically motivated erasure--Pauli channel for polarization-qubit distribution, in which photon loss is represented as an erasure event and residual polarization noise as a Pauli channel, with detector dark counts converting a fraction of erasures into spurious, randomizing clicks. This model directly connects measurable fiber parameters to the effective single-qubit channel, but it was formulated and validated for a single bipartite link, and does not by itself characterize how the genuinely multipartite coherence of a GHZ state degrades when propagated simultaneously across $N-1$ such links. Consequently, a physically grounded, closed-form characterization of QCKA performance that separates local and global error contributions and quantifies their dependence on both distance and the number of participating users remains absent from the current literature.

In this work, we develop a physically grounded finite-key framework for $N$-BB84 based QCKA protocol over
dark-count-modified erasure--Pauli fiber channels and optimize the parameter-estimation
probability to characterize the SKR and secure transmission distance. The main contributions of this work are summarized as follows:
\begin{itemize}
    \item We propose an $N$-party QCKA system based on the $N$-BB84 protocol operating over $N-1$ independent dark-count-modified erasure--Pauli fiber channels, explicitly propagating the shared GHZ state through this physical channel model.
    \item We derive closed-form, channel-parameter-dependent expressions for the local $Z$-basis bit-error rate between Alice and each Bob and for the global $X$-basis phase-error rate associated with the multipartite GHZ coherence, and show that the two are governed by structurally distinct physical mechanisms.
    \item We integrate these physically derived error rates into the composable finite-key secret-key-length framework for $N$-BB84, obtaining finite-key and asymptotic SKR expressions in terms of fiber attenuation, PMD coefficient, photon bandwidth, and dark-count probability.
    \item We formulate the allocation of the parameter-estimation probability as a one-dimensional optimization problem and solve it via Brent's derivative-free method, quantifying the resulting improvement in SKR and maximum secure distance over a fixed-allocation baseline.
    \item We numerically evaluate the resulting QBER and SKR for a varying number of participating users over a range of fiber distances, characterizing how the achievable secure transmission distance reduces as the number of QCKA
    participants increases.
\end{itemize}


The remainder of this paper is organized as follows. Section~\ref{sec:system}
presents the system model and channel framework, while
Section~\ref{sec:Multipartite_BB84_Protocol} describes the multipartite
$N$-BB84 protocol and key-rate formulation. Section~\ref{sec:math_analysis}
derives the channel-specific QBER expressions, and
Section~\ref{sec:optimization} presents the finite-key SKR optimization
using Brent's method. Section~\ref{sec:simulation_results} discusses the
numerical results, and Section~\ref{sec:conclusion} concludes the paper
and discusses future research directions.



\begin{figure*}[!t]
    \centering
    \includegraphics[width=0.9\textwidth]{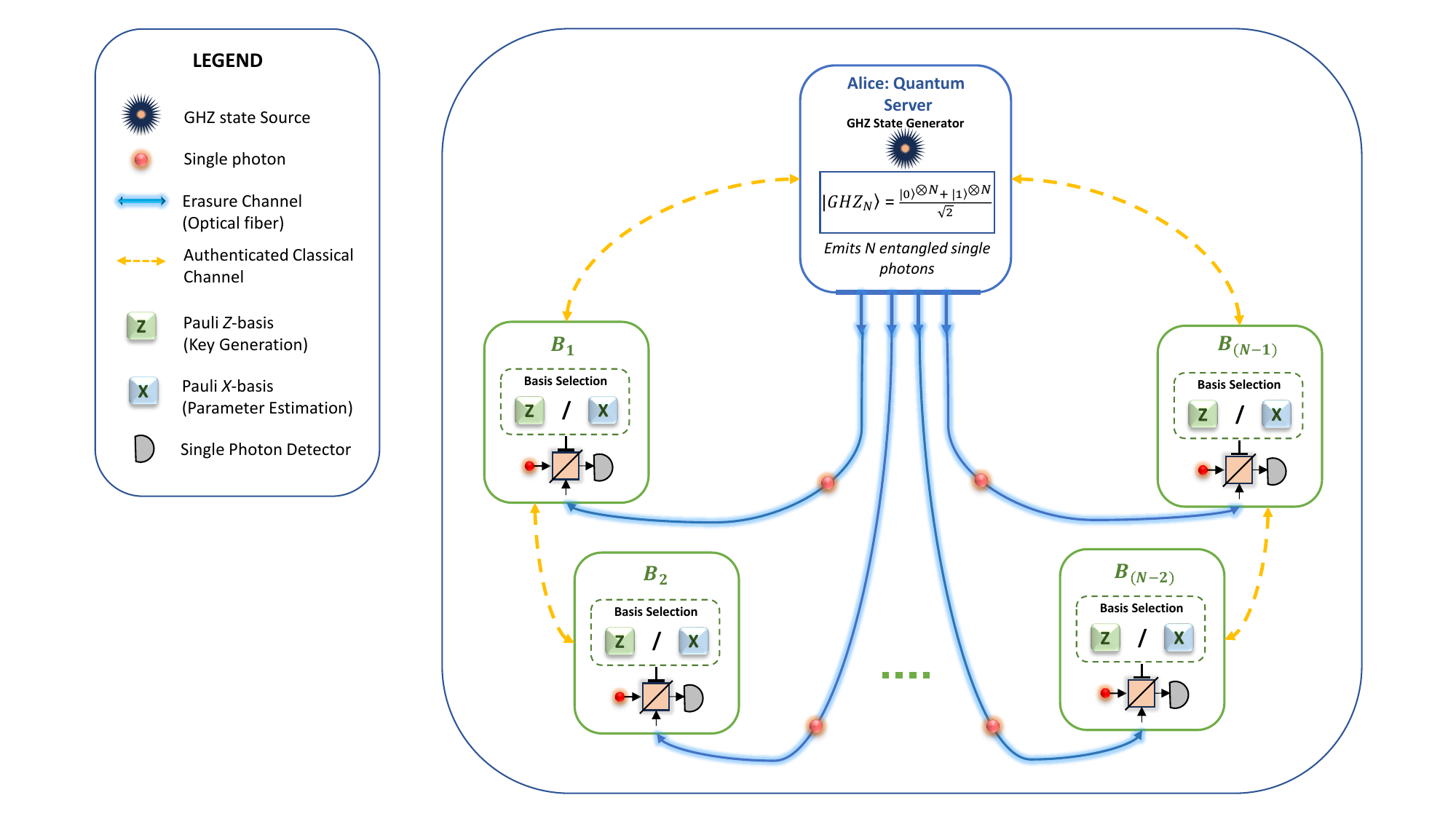}
    \caption{System model of the proposed $N$-party GHZ-state-based quantum conference key agreement system.}
    \label{fig:system_model}
\end{figure*}


\section{System Model}
\label{sec:system}

As depicted in Fig.~\ref{fig:system_model}, we consider an $N$-party QCKA system in which an $N$-qubit GHZ state is distributed from Alice to $N-1$ receivers named as Bob$_i$ through independent optical-fiber links.  During transmission, the distributed qubits experience fiber loss, polarization noise, while detector dark counts are introduced at the
receiver detector, resulting in spurious detection events in the
absence of a received photon. The dark-count effect is
therefore incorporated separately into the channel model as a
dark-count-modified erasure--Pauli channel.
The resulting noisy GHZ state is then used in the $N$-BB84 protocol, where $Z$-basis measurements generate the raw conference key and $X$-basis measurements are used for parameter estimation. The following subsections describe the GHZ state, the physical fiber channel, and the protocol in detail.

\subsection{GHZ State}
\label{Ghz state}
Multipartite QKD enables a single sender, Alice $(A)$ to establish correlated secret keys with multiple Bobs $(\{B_i\}_{i=1}^{N-1})$ via shared entanglement. The $N$-qubit GHZ state vector is given by
\begin{equation}
|\mathrm{GHZ}_N\rangle = \frac{1}{\sqrt{2}}\left(|0\rangle^{\otimes N} + |1\rangle^{\otimes N}\right)\;,
\label{eq:ghz}
\end{equation}
and the resulting density matrix can be expressed as
\begin{equation}
\rho_{\mathrm{GHZ}} = |\mathrm{GHZ}_N\rangle\langle\mathrm{GHZ}_N| = \frac{1}{2}\sum_{k=1}^{4} T_k,
\label{eq:ghzrho}
\end{equation}
where
\begin{align}
T_1 &= \ketbra{0}{0}_A \otimes \ketbra{0}{0}_{B_1}\otimes\cdots\otimes\ketbra{0}{0}_{B_{N-1}},\\
T_2 &= \ketbra{0}{1}_A \otimes \ketbra{0}{1}_{B_1}\otimes\cdots\otimes\ketbra{0}{1}_{B_{N-1}},\\
T_3 &= \ketbra{1}{0}_A \otimes \ketbra{1}{0}_{B_1}\otimes\cdots\otimes\ketbra{1}{0}_{B_{N-1}},\\
T_4 &= \ketbra{1}{1}_A \otimes \ketbra{1}{1}_{B_1}\otimes\cdots\otimes\ketbra{1}{1}_{B_{N-1}}.
\end{align}
The diagonal terms $T_1, T_4$ govern the $Z$-basis bit-error statistics, while the
off-diagonal terms $T_2, T_3$ carry the $N$-partite coherence that determines the
global phase-error rate $Q_X$.

\subsection{Dark-Count-Modified Erasure--Pauli Channel}
\label{ssec:erasure_gen}

We model the loss and detector noise in fiber-based distribution of
polarization-entangled qubits using the erasure--Pauli channel
introduced in~\cite{pirandola2026}. For a single link with
transmissivity $\eta$, the channel acts on an input polarization qubit
$\rho$ as
\begin{equation}
\mathcal{E}_{\mathrm{EP}}(\rho)=(1-\eta)\ketbra{e}{e}+\eta\,\mathcal{P}(\rho),
\label{eq:ep_basic}
\end{equation}
where $\ket{e}$ is an orthogonal vacuum state flagging photon loss, and
$\mathcal{P}$ is a qubit Pauli channel with Kraus representation~\cite{Nielsen_Chuang_2010}
\begin{equation}
\mathcal{P}(\rho)=\sum_{k=0}^{3}p_k\,\sigma_k\rho\sigma_k^\dagger,\label{eq:pauli_channel_def}
\end{equation}
where $\mathbf{p}:=\{p_k\}$ denotes a probability distribution over the Pauli
errors $\sigma_k\in\{I,X,Y,Z\}$, with $I$ denoting the $2\times2$ identity operator.

Applying this channel to polarization-entanglement distribution over standard
optical fiber where PMD is the dominant
noise source~\cite{poole1986,wai1996,Agrawal2012},
gives rise to two different regimes: a depolarizing-dominated regime restricted
to short distances and a dephasing-dominated regime, enabled by active
polarization control, that permits much longer
transmission distances ~\cite{Massar_2007}. We consider the latter case throughout this work, in
which the Pauli channel reduces to a pure, isotropic dephasing in the
$Z$ basis, i.e., $\{p_k\}$ takes the restricted form
\begin{equation}
p_0 = 1-p_Z,\quad p_1=p_2=0,\quad p_3=p_Z,
\label{eq:dephasing_probs}
\end{equation}
so that $\mathcal{P}$ reduces to
\begin{equation}
\mathcal{P}(\rho)=(1-p_Z)\rho+p_Z Z\rho Z,
\end{equation}
with dephasing probability
\begin{equation}
p_Z(d)=\tfrac12\big(1-e^{-d/L_{\mathrm{DH}}}\big)
\label{eq:pz}
\end{equation}
where $d$ denotes the fiber length
and $L_{\mathrm{DH}}$ denotes the decoherence length given by
\begin{equation}
L_{\mathrm{DH}}=2\tau^2/D_{\mathrm{PMD}}^2,
\label{eq:LDH}
\end{equation}
where $\tau \simeq (2\pi\Delta\nu)^{-1}$ is the photon coherence time with $\Delta\nu$ being the bandwidth, and $D_{\mathrm{PMD}}$ is the PMD coefficient in ps/$\sqrt{\mathrm{km}}$.  In telecom fibers $D_{\mathrm{PMD}} \sim 0.01$--$10$ ps/$\sqrt{\mathrm{km}}$ ~\cite{Agrawal2012}. So, for a bandwidth $\Delta\nu = 100$ GHz ($\tau \simeq 1.6$ ps), the decoherence length is in the range of $L_{\mathrm{DH}} \sim 0.05$--$5\times10^{4}$ km. 


To capture realistic detection imperfections, we further incorporate
detector dark counts where even a no-photon (erasure) event can trigger a
spurious detector click, recorded as a random polarization outcome.
This effectively converts a fraction of the erasure events into
depolarizing noise, raising the overall Pauli error probability from
the bare distribution $\{p_k\}$ of Eq.~\eqref{eq:pauli_channel_def} to
an effective distribution $\{p_k'\}$ over the same Pauli operators
$\sigma_k$, yielding the effective channel $\mathcal{P}'$, of the same
form as $\mathcal{P}$ in Eq.~\eqref{eq:pauli_channel_def} with $p_k$
replaced by $p_k'$. This induces the dark-count-modified channel \cite{pirandola2026}
\begin{equation}
\mathcal{E}_{EP}^{dc}(\rho) = (1-\eta')\ketbra{e}{e} + \eta' \mathcal{P}'(\rho),
\label{eq:eeq_gen}
\end{equation}
where the effective transmissivity is
\begin{equation}
\eta' = \eta+(1-\eta)p_{\mathrm{dc}} = 1-(1-\eta)(1-p_{dc}), \label{eq:etap}
\end{equation}
where $p_{dc}$ denotes the dark count probability and the effective Pauli distribution $\{p_k'\}$ admits
\begin{align}
p'_k &= r\,p_k + \frac{1-r}{4}, \label{eq:p_prime}\\
r &= \frac{\eta}{\eta'}, \label{eq:r}
\end{align}
with single-link transmissivity $\eta=10^{-\alpha d/10}$, where $d$ is
the fiber distance in km and $\alpha$ is the attenuation coefficient in dB/km.

\subsection{Measurement and Basis Selection}
\label{subsec:measurement}

To describe the measurement procedure in a general form, consider first
an arbitrary single-qubit state $\rho$. A projective measurement in the
$Z$ basis is described by the POVM elements
\begin{equation}
\Pi^0_Z = \ketbra{0}{0}, \qquad \Pi^1_Z = \ketbra{1}{1},
\label{eq:meas_z_povm}
\end{equation}
while an $X$-basis measurement is described by
\begin{equation}
\Pi^0_X = \ketbra{+}{+}, \quad \Pi^1_X = \ketbra{-}{-},
\quad \ket{\pm} = \frac{\ket{0}\pm\ket{1}}{\sqrt2}.
\label{eq:meas_x_povm}
\end{equation}

For an $N$-party state $\rho$ shared among Alice and the $N-1$ Bobs,
measurements are performed jointly by applying the corresponding
single-qubit POVM element to each party's subsystem. When all $N$
parties select the same basis $M$, the outcome is a bit string
$\bm{a}=(a_A,a_{B_1},\dots,a_{B_{N-1}})\in\{0,1\}^N$, obtained with joint
probability\cite{Gianfranco}
\begin{equation}
P(\bm{a}\mid M) = \Tr\!\big[\Pi^{\bm{a}}_M\,\rho\big],
\label{eq:meas_born_joint}
\end{equation}
where 
\begin{equation}
    \Pi^{\bm{a}}_M = \Pi^{a_A}_M \otimes \Pi^{a_{B_1}}_M \otimes \cdots \otimes \Pi^{a_{B_{N-1}}}_M
\end{equation}
and the corresponding post-measurement state is
\begin{equation}
\rho^{(M)}_{\bm{a}} = \frac{\Pi^{\bm{a}}_M\,\rho\,\Pi^{\bm{a}}_M}{P(\bm{a}\mid M)}.
\label{eq:meas_postmeasurement_joint}
\end{equation}
In particular, a $Z^{\otimes N}$ measurement corresponds to $M=Z$ in
\eqref{eq:meas_born_joint}, and an $X^{\otimes N}$ measurement
corresponds to $M=X$; these are the two joint measurements used
throughout this work, for key generation and parameter estimation,
respectively.

Two correlators built from these joint measurements will be used
for the QBER analysis as described next. The first is the two-party $Z$-basis
correlator between Alice and an individual Bob $B_i$, obtained by
marginalizing \eqref{eq:meas_born_joint} (with $M=Z$) over all parties
other than $A$ and $B_i$:
\begin{align}
\langle Z_A Z_{B_i}\rangle
&= \sum_{a_A,a_{B_i}} (-1)^{a_A\oplus a_{B_i}}\, P(a_A,a_{B_i}\mid Z,Z)\nonumber\\
&= \Tr\!\big[\rho_{AB_i}(Z\otimes Z)\big],
\label{eq:zz_correlator_def}
\end{align}
where $\rho_{AB_i}$ is the corresponding bipartite reduced state. The
second is the full $N$-party $X$-basis correlator,
\begin{align}
\langle X^{\otimes N}\rangle
&= \sum_{\bm{a}} (-1)^{\,a_A\oplus a_{B_1}\oplus\cdots\oplus a_{B_{N-1}}}\, P(\bm{a}\mid X)\nonumber\\
&= \Tr\!\big[\rho\, X^{\otimes N}\big],
\label{eq:xn_correlator_def}
\end{align}
which depends on the full $N$-party outcome parity and
therefore on the complete state $\rho$. Both correlators take values in
$[-1,1]$, with $\langle Z_A Z_{B_i}\rangle=\langle X^{\otimes N}\rangle=1$
in the ideal, noiseless case and deviations from unity are quantified by
the error rates $Q_{AB_i}$ and $Q_X$ introduced in the following
section.

\section{Multipartite BB84 QCKA Protocol}
\label{sec:Multipartite_BB84_Protocol}
QCKA extends QKD from two parties to $N>2$ parties sharing a single secret key, rather than $N-1$ pairwise keys, by replacing the bipartite entangled resource with a multipartite one (e.g., the $N$-qubit GHZ state) and adapting BB84's sifting-and-testing structure~\cite{murta2020quantum}. This yields the $N$-BB84 protocol, where Alice and $N-1$ Bobs measure their GHZ shares and combine local bit statistics with a global check to certify security~\cite{Grasselli_2018}. We describe the protocol below, then define the QBER and phase-error rate used for parameter estimation, and derive the secret key length in the finite-key and asymptotic regimes.

\subsection{$N$-BB84 Protocol}
\label{ssec:protocol}

The $N$-BB84 protocol is the direct $N$-partite generalization of the bipartite BB84 protocol, first proposed in~\cite{Grasselli_2018}. It involves a single sender, Alice ($A$), and $N-1$ receivers, the Bobs ($B_1,\dots,B_{N-1}$), who aim to establish a common secret conference key from a shared multipartite entangled state (e.g., a GHZ state) ~\cite{murta2020quantum}. The protocol consists of the following steps:

\begin{enumerate}
    \item \textit{Distribution.} In each of $L$ rounds, a source distributes one copy of the $N$-qubit GHZ state
    \[
    \ket{\mathrm{GHZ}_N} = \frac{1}{\sqrt2}\big(\ket{0}^{\otimes N} + \ket{1}^{\otimes N}\big)
    \]
    among the $N$ parties: one qubit to Alice and one qubit to each Bob $B_i$. Prior to this quantum distribution, the parties also share a private random sequence $t=(t_1,\dots,t_L)\in\{Z,X\}^L$, i.i.d. with $\Pr[t_k=X]=p$ and $\Pr[t_k=Z]=1-p$, generated e.g.\ by a trusted source or by pre-distributed correlated randomness. This sequence must remain unknown to any adversary until after round $k$'s qubits have been measured.

    \item \textit{Measurement:} In round $k$, \emph{every} party measures its qubit in the common basis $t_k$ dictated by the pre-shared sequence, rather than choosing bases independently at random. Since the basis is synchronized across all $N$ parties, every round is directly usable, there is no basis-mismatch waste, and the $L$ rounds split approximately into $m=Lp$ test rounds ($t_k=X$) and $L-m$ key-generation rounds ($t_k=Z$).

    \item \textit{Parameter Estimation:} The parties use the $m$ test rounds ($t_k=X$) and an additional, disjoint sample of $m$ bits drawn from the $Z$-basis rounds for each Bob $B_i$ to estimate the two error rates:
    \begin{itemize}
        \item \emph{Global phase-error rate $Q_X^m$:} on the $m$ $X$-basis rounds, the parties publicly announce their outcomes (the basis is already common knowledge among them) and check the $N$-partite parity, giving $Q_X^m$.
        \item \emph{Bipartite bit-error rate $Q_{AB_i}^m$:} from the $Z$-basis rounds, Alice and each Bob $B_i$ publicly reveal a random sample of $m$ of their bits and directly compare them, giving $Q_{AB_i}^m$.
    \end{itemize}
    Since the basis choice was pre-shared rather than independently randomized, no separate public basis-reconciliation step is needed and only the announced outcomes above are revealed. The $Z$-basis bits consumed for the $Q_{AB_i}^m$ sample are discarded from the raw key, along with the $X$-basis test rounds, leaving $n=L-2m$ rounds for key generation.

    \item \textit{Raw key:} The remaining $n=L-2m$ key-generation rounds provide the raw key material with Alice's string $A^n$ and each Bob's string $B_i^n$.

    \item \textit{Classical post-processing:} Alice performs one-way error correction (EC) with each Bob in turn, so that every $B_i^n$ is reconciled with $A^n$; all parties then apply privacy amplification (PA), using the same two-universal hash function on their corrected data, to distill a secret conference key of length $\ell<n$ shared by $A$ and all the Bobs.
\end{enumerate}


An important structural feature of $N$-BB84 (and of GHZ-based CKA more generally) is that, for $N\ge 3$, perfect correlation among all $N$ outcomes can only be certified in the $Z$ basis. Unlike the bipartite case, no single choice of local bases yields perfect correlation in a second, independent basis~\cite{murta2020quantum}. This is why two structurally different error rates are needed, i.e., the $N-1$ pairwise bit-error rates $Q_{AB_i}^m$ and a single collective phase-error rate $Q_X^m$, rather than the single QBER of standard bipartite BB84.

\subsection{QBER} 
\label{ssec:qber}
The security of $N$-BB84 is parametrized by two error rates, one certifying the raw key's \emph{correctness} and the other its \emph{secrecy}, both estimated from the $m$-round test sample described in the parameter estimation step. Following the finite-key framework of~\cite{Grasselli_2018}, the raw key is
extracted in the $Z$ basis. Its \emph{correctness} is certified by the
$Z$-basis pairwise QBER between Alice and each Bob, given by
\begin{equation}
Q_{AB_i}^m = \frac{1-\langle Z_A Z_{B_i}\rangle_m}{2},
\quad
\langle Z_A Z_{B_i}\rangle = \Tr\!\big[\rho_{AB_i}(Z\otimes Z)\big],
\label{eq:QABi_}
\end{equation}
where $Z_A Z_{B_i}$ is averaged over $m$ 1st-type rounds randomly chosen by Alice, and $\rho_{AB_i}$ is the reduced bipartite state shared by Alice and Bob $i$ (obtained by tracing out all other parties). It is to be noted that
$Q_{AB_i}$ is a purely \textit{local} quantity.

The \emph{secrecy} of the raw key is certified by the global $N$-party $X$-basis parity
error rate, given by
\begin{equation}
Q_X^m = \frac{1-\langle X^{\otimes N}\rangle_m}{2},
\quad
\langle X^{\otimes N}\rangle = \Tr\!\big[\rho_{AB}\, X^{\otimes N}\big],
\label{eq:QX_}
\end{equation}
where $X^{\otimes N}$ is averaged over the 2nd-type rounds, and $\rho_{AB}$ is the full $N$-party post-selected state shared by Alice and all the Bobs. Unlike $Q_{AB_i}, Q_X$ cannot be assigned to any single Bob, which is a collective property of the whole shared state. In the ideal
noiseless case, $Q_{AB_i}=Q_X=0$.

The frequencies $Q_{AB_i}^m$ and $Q_X^m$ observed in the finite-size test sample are used to bound the eavesdropper's information. The bit error rates $Q_{AB_i}$ enter the error-correction step (since Alice reconciles her string separately with each Bob), while the phase error rate $Q_X$ alone determines how much privacy amplification is needed to erase Eve's information about the joint key.
\subsection{Finite Size Secret Key Length}
\label{sec:finite-key-rate}
We recall Theorem 3 of~\cite{Grasselli_2018}, which gives a fully composable,
finite-size secret key length for $N$-BB84, obtained via the uncertainty
relation for smooth entropies~\cite{UncertaintyRelation}.
Let the $N$-BB84 protocol use the optimal one-way EC protocol, which is
$\varepsilon_{\mathrm{EC}}$-fully secure and $2(N-1)\varepsilon_{\mathrm{PE}}$-robust.
Then the secret key produced by two-universal hashing has length
\begin{align}
\mathrm{SKR} = {} &
\frac{n}{L}
\Bigg[
1 - h\Big(Q_X^m + 2\xi(\varepsilon_x,n,m)\Big)
\notag\\
&\qquad
- \max_i h\Big(Q_{AB_i}^m + 2\xi(\varepsilon_z,n,m)\Big)
\Bigg]
\notag\\
&-\frac{1}{L}\log_2
\frac{2(N-1)}{\varepsilon_{\mathrm{EC}}}
\notag\\
&-\frac{2}{L}\log_2
\frac{1-2(N-1)\varepsilon_{\mathrm{PE}}}
{2\varepsilon_{\mathrm{PA}}}
-\frac{h(p)}{L}.
\label{eq:theorem3}
\end{align}
and is $\varepsilon_{\mathrm{tot}}$-secure, with
\begin{equation}
\varepsilon_{\mathrm{tot}} = 2\varepsilon_{\mathrm{PE}} + \varepsilon_{\mathrm{EC}} + \varepsilon_{\mathrm{PA}}.
\end{equation}
Here, the parameter-estimation error term is
\begin{equation}
\varepsilon_{\mathrm{PE}} \equiv \sqrt{(N-1)\varepsilon_z + \varepsilon_x},
\label{eq:epsPE}
\end{equation}
and the finite-sampling correction is
\begin{equation}
\xi(\varepsilon, n, m) \equiv \sqrt{\frac{(n+m)(m+1)}{8nm^2}\ln\!\left(\frac{1}{\varepsilon}\right)}.
\label{eq:xi}
\end{equation}
Here $h(\cdot)$ denotes the binary entropy function, and the security
parameters $\varepsilon_x, \varepsilon_z, \varepsilon_{\mathrm{EC}},
\varepsilon_{\mathrm{PA}} > 0$ are free to choose subject to the fixed total
security budget $\varepsilon_{\mathrm{tot}}$. Each term in Eq.~\eqref{eq:theorem3} has a direct operational meaning:
\begin{itemize}
    \item The bracketed term is the extractable key rate after the privacy-amplification cost $h(Q_X^m + 2\xi(\varepsilon_x,n,m))$, set by the global phase-error rate $Q_X^m$ from Eq.~\eqref{eq:QX_}, and the error-correction cost $\max_i h(Q_{AB_i}^m + 2\xi(\varepsilon_z,n,m))$, set by the worst pairwise bit-error rate $Q_{AB_i}^m$ from Eq.~\eqref{eq:QABi_}, both evaluated on the $m$ test rounds.
    \item $n = L-2m$ is the number of key-generation rounds surviving after test rounds and sampling.
    \item The $\xi(\varepsilon,n,m)$ terms are finite-sampling corrections that bound the deviation between the observed test-round frequencies and the true error rates on the remaining $n$ rounds.
    \item The two $\log_2(\cdot)$ terms account for information leaked during error correction and privacy amplification, respectively.
    \item $h(p)$ is the fraction of the final key sacrificed to account for the pre-shared randomness used to mark the test (2nd-type) rounds.
\end{itemize}

\subsection{Asymptotic Key Rate}
\label{sec:Asymptotic-key-rate}

In the asymptotic limit $L\to\infty$, the finite-size corrections vanish
($\xi\to 0$, and the $\log_2(\cdot)$ and $h(p)$ terms become negligible relative
to $n$), and the key rate per signal reduces to
\begin{equation}
r_{BB84}(P_{AB}) = 1 - h(Q_X) - \max_i h(Q_{AB_i}),
\label{eq:asymptotic_rate}
\end{equation}
which recovers the multipartite BB84 rate derived in~\cite{murta2020quantum}. As in the finite-key case, the two terms retain their operational roles: $h(Q_X)$ is the privacy-amplification cost against a collective attack consistent with the observed global phase-error rate, while $\max_i h(Q_{AB_i})$ is the error-correction cost of reconciling with the noisiest Bob.

\section{QBER Analysis for Fiber Optic Channel}
\label{sec:math_analysis}

We now specialize the general $N$-BB84 security analysis of Sec.~\ref{sec:Multipartite_BB84_Protocol} to the physical channel model introduced in Sec.~\ref{sec:system}, i,.e., the dark-count--erasure-Pauli channel. We assume that this channel acts independently on each Bob's qubit, in the \emph{homogeneous} setting where all $N-1$ Bobs are equidistant from the source, i.e., Alice. Starting from the channel definition, the dark-count-modified transmissivity and the effective Pauli
probabilities given in Eqs.~\eqref{eq:etap}--\eqref{eq:p_prime}
and considering the dephasing-dominated regime described by
Eqs.~\eqref{eq:pz}--\eqref{eq:LDH}, we derive its action on the computational-basis operators, the resulting post-selected GHZ state, and its bipartite marginals, and integrate into the dark count enabled Pauli channel in Eqn.~\ref{eq:eeq_gen} with the
$N$-qubit GHZ state defined in Section~\ref{Ghz state}. The resulting output state before post-selection is given by
\begin{equation}
\rho_{\mathrm{out}}
=
\left(
\mathcal{I}_{A}
\otimes
\eop_{EP}^{dc\otimes(N-1)}
\right)
(\rho_{\mathrm{GHZ}}),
\label{eq:rho_out}
\end{equation}
where $\mathcal{I}_{A}$ denotes the identity operation on Alice's qubit
and $\mathcal{E}^{\mathrm{dc}}_{EP}$ is the dark-count-modified
erasure--Pauli channel defined in \eqref{eq:eeq_gen}.


\subsection{Channel Action on the GHZ State}
\label{sec:channel_action}

We now propagate the $N$-qubit GHZ density matrix $\rho_{\mathrm{GHZ}}$  from \eqref{eq:ghzrho} through the physical channel of Sec.~\ref{ssec:erasure_gen}. The effective Pauli channel of \eqref{eq:pauli_channel_def} acts
on a single-qubit operator $\rho$ as
\begin{equation}
\mathcal{P}'(\rho) = \sum_{k=0}^{3} p_k' \,\sigma_k \rho\, \sigma_k^{\dagger}.
\label{eq:Pprime-def}
\end{equation}
Using the standard Pauli-conjugation identities
$$X\ketbra00 X=Y\ketbra00 Y=\ketbra11,\; Z\ketbra00 Z=\ketbra00$$ (and the
analogous relations with $0\leftrightarrow1$ exchanged), together with
$$X\ketbra01 X=\ketbra10=-Y\ketbra01 Y,\; Z\ketbra01 Z=-\ketbra01,$$
direct evaluation of \eqref{eq:Pprime-def} on the four operators
$\ketbra00,\ketbra11,\ketbra01,\ketbra10$ spanning $\rho_{\mathrm{GHZ}}$ gives
\begin{align}
\mathcal{P}'(\ketbra00) &= a_0\ketbra00 + a_1\ketbra11,\nonumber \\
\mathcal{P}'(\ketbra11) &= a_1\ketbra00 + a_0\ketbra11, \nonumber\\
\mathcal{P}'(\ketbra01) &= b_0\,\ketbra01 + b_1\,\ketbra10,\nonumber \\
\mathcal{P}'(\ketbra10) &= b_0\,\ketbra10 + b_1\,\ketbra01,\nonumber \label{eq:P01-P10}
\end{align}
with $a_0=p_0'+p_3'$, $a_1=p_1'+p_2'$, $b_0=p_0'-p_3'$, $b_1=p_1'-p_2'$; the
normalization $\{p_k'\}$ ensures $a_0+a_1=1$. Specializing to the dephasing-dominated regime, from \eqref{eq:dephasing_probs}
$$p_0=1-p_Z,p_1=p_2=0, p_3=p_Z,$$ and substituting into ~\eqref{eq:p_prime} gives 
\begin{equation}
a_0=\frac{1+r}{2}, \hspace{0.5cm} a_1=\frac{1-r}{2},
\end{equation}
\begin{equation}
    b_0=r(1-2p_Z)=\frac{\eta}{\eta'}(1-2p_Z), \hspace{0.5cm} b_1=0.
\label{eq:a0a1cd-final}
\end{equation}
The vanishing of $b_1$ is the key structural fact for everything that
follows: under pure dephasing, $\ketbra01$ and $\ketbra10$ never mix
into one another, for any value of $\eta$, $p_{dc}$, or $p_Z$.

Including the erasure flag, the full single-qubit channel from~\eqref{eq:eeq_gen} therefore maps the four basis operators to
\begin{align}
\mathcal{E}_{EP}^{dc}(\ketbra00) &= (1-\eta')\ketbra ee + \eta'\sigma_0,\nonumber\\
\mathcal{E}_{EP}^{dc}(\ketbra11) &= (1-\eta')\ketbra ee + \eta'\sigma_1, \notag\\
\mathcal{E}_{EP}^{dc}(\ketbra01) &= (1-\eta')\ketbra ee + \eta' b_0\,\ketbra01,\nonumber\\
\mathcal{E}_{EP}^{dc}(\ketbra10) &= (1-\eta')\ketbra ee + \eta' b_0\,\ketbra10,
\label{eq:Lambda-all}
\end{align}
where $b_1$ has already removed the cross term, and
$$\sigma_0:=a_0\ketbra00+a_1\ketbra11, \sigma_1:=a_1\ketbra00+a_0\ketbra11$$
each have unit trace. Applying $\mathcal{I}_A\otimes(\mathcal{E}_{EP}^{dc})^{\otimes(N-1)}$
term-by-term to $T_1,\dots,T_4$ maps each onto Alice's untouched qubit
tensored with $(N-1)$ copies of the corresponding image in
\eqref{eq:Lambda-all}.

\subsection{Post-Selected Output State}
\label{sec:output_state}
Key generation is retained only for rounds in which all $N-1$ Bobs register a detection, i.e., no qubit is lost to the erasure flag $\ket e$. The single-qubit detection projector $\Pi_{\det}=\ketbra00+\ketbra11$ annihilates $\ketbra ee$ and acts as the identity on the remaining three operators in \eqref{eq:Lambda-all}, so post-selection simply discards the erasure term from each image. Applying $\Pi_{\det}^{\otimes(N-1)}(\cdot)\Pi_{\det}^{\otimes(N-1)}$ to the mapped $T_k$'s and summing with the $1/2$ prefactor gives the un-normalized, detection-conditioned operator
\small
\begin{align}
&\tilde\rho_{AB} = \frac{1}{2}\Big[
\ketbra00_A \otimes (\eta'\sigma_0)^{\otimes(N-1)}
+ \ketbra11_A \otimes (\eta'\sigma_1)^{\otimes(N-1)} \notag\\
&+ (\eta' b_0)^{N-1}\Big(\ketbra01_A \otimes \ketbra01_B^{\otimes(N-1)}
+ \ketbra10_A \otimes \ketbra10_B^{\otimes(N-1)}\Big)\Big].
\label{eq:rho-tilde}
\end{align}
\normalsize
Since $\Tr[\sigma_0]=\Tr[\sigma_1]=1$ while $\Tr[\ketbra01]=\Tr[\ketbra10]=0$, only the diagonal terms contribute to the trace, giving a joint detection probability of $P_{\mathrm{succ}}=(\eta')^{N-1}$. Crucially, the coherence (off-diagonal) term in \eqref{eq:rho-tilde} does \emph{not} vanish as an operator merely because it is traceless: it survives with coefficient $(\eta' b_0)^{N-1}$ and must be kept when normalizing, even though it contributes nothing to $P_{\mathrm{succ}}$ itself. Dividing \eqref{eq:rho-tilde} by $P_{\mathrm{succ}}=(\eta')^{N-1}$, the diagonal terms lose their $(\eta')^{N-1}$ prefactor exactly, while the coherence term reduces from $(\eta'b_0)^{N-1}$ to $b_0^{N-1}$, since $(\eta'b_0)^{N-1}/(\eta')^{N-1}=b_0^{N-1}$. This resolves what would otherwise look like an ambiguity between $(\eta'b_0)^{N-1}$ and $b_0^{N-1}$ as the physical coherence factor: the correct, normalized
post-selected state is
\small 
\begin{equation}
\rho_{AB}^{(\mathrm{ps})} = \frac{\tilde\rho_{AB}}{P_{\text{succ}}}
= \frac{1}{2}\Big[
|0\rangle\langle 0|_A \otimes \sigma_0^{\otimes(N-1)}
+ |1\rangle\langle 1|_A \otimes \sigma_1^{\otimes(N-1)}
\notag
\end{equation}
\begin{equation}
\hspace{2.2em}+ b_0^{N-1}\Big(|0\rangle\langle 1|_A \otimes |0\rangle\langle 1|_B^{\otimes(N-1)}
+ |1\rangle\langle 0|_A \otimes |1\rangle\langle 0|_B^{\otimes(N-1)}\Big)\Big],
\label{eq:rho-ps-final}
\end{equation}
\normalsize
with $b_0$ given by \eqref{eq:a0a1cd-final}. This is the state used
throughout the remainder of this section.

\subsection{Reduced Bipartite State and the Two Error Rates}
\label{sec:reduced_state}

Tracing $\rho_{AB}^{(\mathrm{ps})}$ over every Bob $B_j$, $j\ne i$,
leaves the diagonal terms unchanged, since $\Tr[\sigma_0]=\Tr[\sigma_1]=1$,
but kills the coherence term identically because $\Tr[\ketbra01]=\braket10=0$
and $N\ge3$ guarantees at least one Bob to trace over. Therefore, the off-diagonal
contribution in \eqref{eq:rho-ps-final} vanishes on any bipartite
marginal. The $N$-partite coherence is therefore a genuinely global
property that never survives in any pairwise reduced state consistent
with the structural remark of Sec.~\ref{ssec:protocol} that, for $N\ge3$,
correlation beyond the $Z$ basis cannot be certified from any single
pair of parties. The reduced state is thus
\begin{align}
\rho_{AB_i}& = \frac{1}{2}\left[\ketbra00_A\otimes\sigma_0 + \ketbra11_A\otimes\sigma_1\right]\nonumber\\
&= \frac{1}{2}
\begin{pmatrix}
\frac{1+r}{2} & 0 & 0 & 0\\[2pt]
0 & \frac{1-r}{2} & 0 & 0\\[2pt]
0 & 0 & \frac{1-r}{2} & 0\\[2pt]
0 & 0 & 0 & \frac{1+r}{2}
\end{pmatrix},
\label{eq:rho-ABi-matrix}
\end{align}
in the ordered basis 
$\{\ket{00},\ket{01},\ket{10},\ket{11}\}$, 
with
$\Tr[\rho_{AB_i}]=1$. \\

\begin{theorem}[Bit-error rate]\label{thm:QABi}
Let $N\ge 3$, and let $\rho_{\mathrm{GHZ}}$ be distributed through $N-1$ independent copies of $\mathcal{E}^{\mathrm{dc}}_{EP}$ in the dephasing-dominated regime, post-selected on all $N-1$ detections. The bit-error rate between Alice and any Bob $B_i$ admits
\begin{equation}
Q_{AB_i}=\frac{1-r}{2}=\frac{\eta'-\eta}{2\eta'}=\frac{(1-\eta)\,p_{dc}}{2\big[\eta+(1-\eta)p_{dc}\big]}\;.
\label{eq:QABi-final}
\end{equation}
Thus, $Q_{AB_i}$ depends only on the
fiber transmissivity $\eta$ and the dark-count probability
$p_{\mathrm{dc}}$, and is independent of $p_Z$ and $N$.
\end{theorem}

\begin{IEEEproof}
From the reduced bipartite state in Eq.~\eqref{eq:rho-ABi-matrix},
the $Z$-basis correlation is

\begin{equation}
\begin{aligned}
\langle Z_A Z_{B_i}\rangle
&=
\operatorname{Tr}
\left[
\rho_{AB_i}(Z\otimes Z)
\right] \\
&=\tfrac12\Big[\tfrac{1+r}{2}-\tfrac{1-r}{2}-\tfrac{1-r}{2}+\tfrac{1+r}{2}\Big]\notag\\
&=\tfrac12\big[(1+r)-(1-r)\big]=r .
\end{aligned}
\end{equation}
where $Z\otimes Z=\ketbra{00}{00}-\ketbra{01}{01}-\ketbra{10}{10}+\ketbra{11}{11}$.
Now using the definition
$Q_{AB_i}$ from Eq.~\eqref{eq:QABi_} gives
\begin{equation}
Q_{AB_i}=\frac{1-r}{2}.
\end{equation}
Since $r$ from Eq.~\eqref{eq:r}, and with $\eta'$ from Eq.~\eqref{eq:etap}
it follows that
\begin{equation}
1-r=\frac{\eta'-\eta}{\eta'}=\frac{(1-\eta)p_{dc}}{\eta'},
\label{eq: 1-r}
\end{equation}
Substitution of \ref{eq: 1-r} yields
Eq.~\eqref{eq:QABi-final}. The result is independent of $p_Z$ because pure $Z$-basis
dephasing leaves the $Z$-basis populations unchanged. Hence,
the local bit errors arise solely from the randomization introduced
by dark counts in events associated with photon loss. \\
\end{IEEEproof}


\begin{theorem}[Global phase-error rate]
\label{thm:QX}
Under the assumptions of Theorem~\ref{thm:QABi}, the global $X$-basis phase-error rate is
\begin{equation}
Q_X
=
\frac{1}{2}
\left[
1-
\left(
\frac{\eta e^{-d/L_{\mathrm{DH}}}}
{\eta'}
\right)^{N-1}
\right],
\label{eq:QX-final}
\end{equation}
For $p_{\mathrm{dc}}\ll\eta$, (so $\eta'\approx\eta$) this reduces to
\begin{equation}
Q_X
\approx
\frac{1}{2}
\left[
1-e^{-(N-1)d/L_{\mathrm{DH}}}
\right].
\label{eq:QX-approx}
\end{equation}
\end{theorem}

\begin{IEEEproof}
The global phase-error rate is determined by the
$N$-party correlator in Eq.~\eqref{eq:QX_}.
Since $X^{\otimes N}$ maps
$|0\rangle^{\otimes N}$ to $|1\rangle^{\otimes N}$ and vice versa,
the diagonal terms of the post-selected state
in Eq.~\eqref{eq:rho-ps-final} do not contribute to the correlator. The off-diagonal GHZ coherence is preserved through all $N-1$
links, with one coherence factor contributed by each link.
Moreover, tracing out any Bob removes this off-diagonal contribution;
hence, it appears only in the full $N$-party $X$-basis parity.
Therefore, only the off-diagonal GHZ coherence terms contribute to $\langle X^{\otimes N}\rangle$, via $\bra{0\cdots0}X^{\otimes N}\ket{1\cdots1}=\bra{1\cdots1}X^{\otimes N}\ket{0\cdots0}=1$, giving 
\begin{equation}
\langle X^{\otimes N}\rangle=\tfrac12b_0^{N-1}(1+1)=b_0^{N-1}=\left(\frac{\eta}{\eta'}(1-2p_Z)\right)^{N-1}
\label{eq:xn_erasure}
\end{equation}
where $b_0$ is defined in Eq.~\eqref{eq:a0a1cd-final}.
Therefore,
\begin{equation}
Q_X
=
\frac{1}{2}
\left[
1-
\left(
\frac{\eta}{\eta'}(1-2p_Z)
\right)^{N-1}
\right].
\end{equation}
For the dephasing model from Eq.~\eqref{eq:pz}, $(1-2p_Z)$ can be simplified as
\begin{align}
1 - 2p_Z = 1 - 2 \cdot \frac{1}{2}(1 - e^{-d/L_{DH}}) = e^{-d/L_{DH}} \label{eq:1_minus_2pz}
\end{align}
Therefore
Using Eq.~\eqref{eq:1_minus_2pz},
yields Eq.~\eqref{eq:QX-final}. When
$p_{\mathrm{dc}}\ll\eta$, we have $\eta'\approx\eta$, so $r\to1$ and $b_0\to1-2p_Z$, which
directly yields Eq.~\eqref{eq:QX-approx}.
\end{IEEEproof}
\vspace{0.5cm}
Table~\ref{tab:limiting_checks} verifies that Theorems~\ref{thm:QABi} and~\ref{thm:QX} behave as physically required in every extreme cases, confirming that dark counts affect only $Q_{AB_i}$ and enter $Q_X$ only through the ratio $\eta/\eta'$, which is negligible for $p_{dc}\ll\eta$, while dephasing PMD and the number of Bobs affect only $Q_X$.

\begin{table}[!htbp]
\centering
\caption{Limiting checks of the derived bit- and phase-error rates.}
\label{tab:limiting_checks}
\scriptsize
\renewcommand{\arraystretch}{1.15}
\setlength{\tabcolsep}{3pt}
\begin{tabular}{c|c|p{3.5cm}}
\hline
\textbf{Limit} & \textbf{Result} & \textbf{Physical interpretation} \\
\hline
$p_{dc}\to0$ & $Q_{AB_i}\to0$ & No dark counts; no bit errors. \\
$p_{dc}\to1$ & $Q_{AB_i}\to(1-\eta)/2$ & Lost photons yield fully random outcomes. \\
$\eta\to1$ & $Q_{AB_i}\to0$ & Ideal transmission. \\
$\eta\to0$ & $Q_{AB_i}\to1/2$ & Complete randomization. \\
$d\to0$ & $Q_X\to0$ & Ideal transmission ($\eta=\eta'=1$, $p_Z=0$). \\
$\eta\to0$ & $Q_X\to1/2$ & Complete loss destroys phase correlation. \\
$p_{dc}\to0$ & $Q_X\to\tfrac12\big[1-e^{-(N-1)d/L_{DH}}\big]$ & Only dephasing degrades the phase error. \\
$p_Z\to0$ & $Q_X=\tfrac12\big[1-(\eta/\eta')^{N-1}\big]$ & Only dark-count mixing remains; vanishes as $p_{dc}\to0$. \\
$N\to\infty$ & $Q_X\to1/2$ & GHZ coherence vanishes as each link contributes a further factor below unity. \\
\hline
\end{tabular}
\end{table}

\subsection{Secret Key Rate under the Physical Channel Model}
\label{sec:skr_substitution}


Using the channel-specific bit-error rate $Q_{AB_i}$ and global
phase-error rate $Q_X$ derived in Eqs.~\eqref{eq:QABi-final} and
\eqref{eq:QX-final}, respectively, the finite-key secret key rate is
obtained by substituting these physical error rates into the
general finite-key SKR expression defined in
Eq.~\eqref{eq:theorem3}. For the homogeneous-distance configuration,
all the Bobs experience identical channels, and hence
$\max_i h(Q_{AB_i})$ reduces to $h(Q_{AB_i})$. Similarly, the asymptotic secret key rate is obtained by
substituting the same channel-specific error rates into the
asymptotic SKR expression given in Eq.~\eqref{eq:asymptotic_rate}. Thus, the resulting finite-key and asymptotic rates are directly
determined by the physical fiber parameters, including the
attenuation coefficient $\alpha$, fiber distance $d$, PMD
coefficient $D_{\mathrm{PMD}}$, photon bandwidth $\Delta\nu$,
and dark-count probability $p_{\mathrm{dc}}$.

\section{Secret Key Rate Optimization}
\label{sec:optimization}
In this section, the parameter-estimation probability $p$ is optimized in order to maximize the finite-key SKR. A numerical optimization approach based on Brent's method is applied.
For a fixed number of parties $N$, fiber distance $d$, and total block
length $L$, the finite-key secret key length can be expressed as
\begin{equation}
\ell = \ell(p;N,d,L),
\label{eq:ell_general_p}
\end{equation}
where the dependence on $p$ arises through the parameter-estimation and
key-generation sample sizes given by $m = Lp$ and $n = L-2m = L(1-2p)$.

Among the four quantities $(N,d,L,p)$ that determine $\ell$, only $p$ is
optimized here because it is the only one that is a free \emph{protocol}
parameter rather than a fixed \emph{physical} or \emph{experimental}
setting. The number of parties $N$ and the fiber distance $d$ define the
network topology under study and are treated as external inputs whose
effect on the key rate we wish to characterize; the block length $L$ is
typically fixed by the experiment's duration or hardware budget. The
probability $p$, by contrast, is a design choice made by the parties
themselves at no physical cost, and it directly controls the trade-off
between the two competing resources in Eq.~\eqref{eq:ell_general_p}, where a
larger $p$ increases the test sample $m$ and tightens the statistical
estimate of the error rates $Q_X$ and $Q_{AB_i}$, but simultaneously
shrinks the key-generation sample $n$ from which the raw key is
extracted. Since this trade-off can shift the achievable SKR
substantially, choosing $p$ well is not a minor implementation detail
but a necessary step for reporting a meaningful, optimized key rate. The finite-key SKR is obtained by normalizing the secret key length with
respect to the total number of communication rounds:
\begin{equation}
R_{\mathrm{FK}}(p;N,d,L)
=
\frac{\ell(p;N,d,L)}{L}.
\label{eq:skr_finite_optimization}
\end{equation}

The optimization problem is therefore formulated as
\begin{equation}
\begin{aligned}
p^{\star}(N,d)
&=
\underset{p}{\operatorname{arg\,max}}
\quad
R_{\mathrm{FK}}(p;N,d,L)
\\
\mathrm{subject~to}
\quad&
p_{\min}
\leq p
\leq
p_{\max},
\\
&
m=Lp>0,
\\
&
n=L(1-2p)>0.
\end{aligned}
\label{eq:optimization_problem}
\end{equation}
Here, the notation $p^\star(N,d)$ makes explicit that this optimizing
choice of $p$ depends on the pair $(N,d)$, since both $Q_X$ and
$Q_{AB_i}$, and hence the shape of $R_{\mathrm{FK}}(p)$, depend on $N$
and $d$ through the channel model. Since $n>0$, the physical range of the optimization variable satisfies
\begin{equation}
0<p<\frac{1}{2}.
\label{eq:p_range}
\end{equation}

In the numerical implementation, a slightly restricted interval
\begin{equation}
p_{\min}
\leq
p
\leq
\frac{1}{2}-p_{\min}
\label{eq:numerical_p_range}
\end{equation}
is used instead, to avoid the numerical singularities that occur as
$p\to0$ or $p\to1/2$: the finite-sampling correction
$\xi(\varepsilon,n,m)$ diverges as $m\to0$, and $n\to0$ as $p\to1/2$, so
evaluating $\ell(p;N,d,L)$ exactly at either boundary is numerically
ill-conditioned. The value of $p_{\min}$ is chosen as the smallest
test-sample fraction for which $m=Lp_{\min}$ remains large enough for
$\xi(\varepsilon,n,m)$ to give a statistically meaningful bound; in
practice a small fixed fraction (e.g., $p_{\min}\sim10^{-4}$, or a
minimum test-sample count appropriate to the chosen $L$) is sufficient,
and the resulting optimum $p^\star$ is checked to lie in the interior of
$[p_{\min},\,1/2-p_{\min}]$ rather than at either boundary. For each pair $(N,d)$, the optimization is performed independently.
Thus, varying $N$ or $d$ changes the underlying error parameters $Q_X$
and $Q_{AB_i}$ and consequently modifies the optimal
parameter-estimation probability $p^\star$.

\subsection{SKR Maximization Using Brent's Method}
\label{subsec:brent_optimization}

The optimization problem in \eqref{eq:optimization_problem} is
one-dimensional, since $p$ is the only optimization variable. However,
the dependence of the finite-key secret key length on $p$ is nonlinear,
due to the simultaneous dependence of $m$, $n$, and the statistical
fluctuation terms on $p$:
\begin{equation}
p
\rightarrow
(m,n)
\rightarrow
(\xi_X,\xi_Z)
\rightarrow
(Q_X^{\mathrm{eff}},
Q_{AB}^{\mathrm{eff}})
\rightarrow
\ell.
\label{eq:p_dependency_chain}
\end{equation}
Consequently, a closed-form analytical solution for the optimal $p$ is
generally difficult to obtain directly, and Brent's method is instead
employed to numerically determine the value of $p$ that maximizes the
finite-key SKR.

The original optimization problem is
\begin{equation}
\max_{p}
\;
\ell(p;N,d,L).
\label{eq:max_problem_brent}
\end{equation}
Since standard numerical routines are commonly formulated as minimization
problems, \eqref{eq:max_problem_brent} is equivalently written as
\begin{equation}
\min_{p}
\;
f(p),
\qquad
f(p)
=
-\ell(p;N,d,L).
\label{eq:min_problem_brent}
\end{equation}
The optimum obtained using Brent's method is therefore
\begin{equation}
p^\star
=
\underset{p\in[p_{\min},p_{\max}]}
{\operatorname{arg\,min}}
\left[-\ell(p;N,d,L)\right],
\label{eq:p_opt_brent}
\end{equation}
with corresponding maximum secret key length
\begin{equation}
\ell_{\max}(N,d)
=
\ell(p^\star;N,d,L),
\label{eq:ell_max_brent}
\end{equation}
and optimized finite-key SKR
\begin{equation}
R_{\mathrm{FK}}^{\mathrm{opt}}(N,d)
=
\frac{\ell_{\max}(N,d)}{L}.
\label{eq:optimized_skr}
\end{equation}

Brent's method is well suited to this problem because it does not
require the analytical derivative of the objective function. Instead, it
combines two complementary numerical strategies: golden-section search
and parabolic interpolation \cite{brent2002algorithms}. The golden-section procedure provides a robust mechanism for reducing
the interval containing the optimum. Let the initial optimization
interval be
\begin{equation}
[x_0,y_0]
=
[p_{\min},p_{\max}].
\label{eq:initial_interval}
\end{equation}
At each iteration, the search interval is reduced according to the
relative objective values, so that the uncertainty interval containing
the optimum progressively shrinks,
\begin{equation}
[y_k-x_k]
\longrightarrow
0.
\label{eq:interval_reduction}
\end{equation}
In addition to interval reduction, Brent's method accelerates
convergence using parabolic interpolation. Given three previously
evaluated points,
\begin{equation}
(p_1,f_1),
\quad
(p_2,f_2),
\quad
(p_3,f_3),
\label{eq:three_points}
\end{equation}
a local quadratic approximation of the objective function is constructed,
\begin{equation}
f(p)
\approx
xp^2+yp+c,
\label{eq:parabolic_approximation}
\end{equation}
where $c$ is constant with estimated stationary point
\begin{equation}
p_{\mathrm{par}}
=
-\frac{y}{2x}.
\label{eq:parabolic_minimum}
\end{equation}
If this parabolic step is numerically reliable and remains within the
admissible search interval, it is used to generate the next candidate
point; otherwise, the algorithm reverts to a golden-section step. Brent's
method thus combines the robustness of interval-based search with the
faster local convergence provided by interpolation. The overall procedure
is summarized in Algorithm~\ref{alg:brent_skr}. 

\subsection{Computational Complexity }
\label{subsec:complexity}

For each fixed pair $(N,d)$, Brent optimization approach reduces to a
one-dimensional search over $p$, the complexity required to reach a numerical
tolerance $\delta$ is approximately
\begin{equation}
\mathcal{O}\left(
\log\frac{p_{\max}-p_{\min}}{\delta}
\right).
\label{eq:brent_complexity}
\end{equation}
%
evaluations of $\ell$ (parabolic steps only accelerate this on smooth regions). Since Brent's method optimizes $\ell(p)$ directly, without requiring
derivative evaluations, it is generally more robust for the present
nonlinear finite-key SKR optimization problem. 
\begin{algorithm}[!htbp]
\caption{Brent Optimization of the Finite-Key SKR}
\label{alg:brent_skr}

\begin{algorithmic}[1]
\footnotesize
\Require
$N$, distance set $\mathcal{D}$,
block length $L$, and bounds $[p_{\min},p_{\max}]$

\Ensure
Optimized probability $p^\star$ and SKR $R_{\mathrm{FK}}^{\mathrm{opt}}$

\For{each $N$}

    \For{each fiber distance $d\in\mathcal{D}$}

        \State Define $f(p)=-\ell(p,d,N)$

        \State Set $p_{\min}=10^{-4}$ and
        $p_{\max}=0.5-10^{-4}$

        \State Set $\ell=0$ if
        $Q_X^{\mathrm{eff}}\ge\tfrac12$,
        $Q_{AB_i}^{\mathrm{eff}}\ge\tfrac12$, or
        $1-2(N-1)\varepsilon_{\mathrm{PE}}\le0$;
        otherwise set $\ell\gets\max\{\ell,0\}$

        \State Apply bounded Brent optimization:
        \[
        p^\star=\arg\min_{p\in[p_{\min},p_{\max}]} f(p)
        \]

        \State Compute
        \[
        \ell_{\max}=\ell(p^\star,d,N)
        \]

        \State Compute the optimized SKR:
        \[
        R_{\mathrm{FK}}^{\mathrm{opt}}
        =
        \frac{\ell_{\max}}{L}
        \]

        \State Store $p^\star$ and
        $R_{\mathrm{FK}}^{\mathrm{opt}}$

    \EndFor

\EndFor

\State \Return $p^\star(N,d)$ and
$R_{\mathrm{FK}}^{\mathrm{opt}}(N,d)$

\end{algorithmic}
\end{algorithm}

\section{Numerical Results}
\label{sec:simulation_results}


This section presents the numerical evaluation of the proposed
multipartite GHZ-based quantum communication system under the
erasure-Pauli channel model introduced in Sec.~\ref{sec:system}. The
system's performance is characterized in terms of the QBER and the finite-key SKR, as functions of
the fiber distance and the number of participating users.
All  channel and error-rate quantities ($\eta$, $\eta'$,
$Q_{AB_i}$, $Q_X$) are computed as functions of $N$ and $d$ using the
closed-form expressions derived in Sec.~\ref{sec:math_analysis}, and the
resulting finite-key length $\ell$ and SKR $R_{\mathrm{FK}}=\ell/L$
follow from \eqref{eq:theorem3}. The complete set of numerical parameters used in the simulations is
summarized in Table~\ref{tab:sim_params}.

\begin{table}[!htbp]
\centering
\caption{Numerical parameters used in the simulations.}
\label{tab:sim_params}
\begin{tabular}{ll}
\hline
\textbf{Parameter} & \textbf{Value} \\
\hline
Number of parties, $N$ & $\{3,4,5,6\}$ \\
Total block length, $L$ & $10^{6}$ \\
Parameter-estimation probability, $p$ & $0.02$ \\
Minimum parameter-estimation probability, $p_{\min}$ & $10^{-4}$ \\
Maximum parameter-estimation probability, $p_{\max}$ & $0.5-10^{-4}$ \\
Stopping tolerance of the Brent search, $\delta$ & $10^{-9}$ \\
Security parameters, $\varepsilon_X=\varepsilon_Z=\varepsilon_{\mathrm{EC}}=\varepsilon_{\mathrm{PA}}$ & $10^{-10}$ \\
Fiber attenuation coefficient, $\alpha$ & $0.2$~dB/km \\
Dark-count probability, $p_{\mathrm{dc}}$ & $10^{-3}$ \\
Photon bandwidth, $\Delta\nu$ & $100$~GHz \\
PMD coefficient, $D_{\mathrm{PMD}}$ & $0.1$~ps/$\sqrt{\mathrm{km}}$ \\
\hline
\end{tabular}
\end{table}

Fig.~\ref{fig:qber_distance} plots $Q_{AB_i}$ and $Q_X$ against fiber
distance for each $N$. The two curves confirm, numerically, the
structural distinction predicted analytically in
Sec.~\ref{sec:math_analysis}: $Q_{AB_i}$, given by~\eqref{eq:QABi-final},
depends only on the loss-to-dark-count ratio and is therefore
independent of $N$, saturating toward $1/2$ purely as a function
of distance. Furthermore, $Q_X$, given by~\eqref{eq:QX-final}, instead
carries the coherence factor $b_0^{N-1}$ raised to the $(N-1)$-th power,
so its approach to $1/2$ accelerates visibly as $N$ increases. This is
not merely a numerical observation but a direct consequence of the
$(N-1)$-fold independent-channel structure derived in
Sec.~\ref{sec:math_analysis}. Unlike bipartite QKD, where a single QBER
saturates at a rate set by one link, the global $N$-partite coherence
of a GHZ state decays as the product of $N-1$ per-link survival
factors, each strictly less than unity. Physically, this means that
dark counts and dephasing are not interchangeable 'noise' in this
protocol; they act on different observables and combine differently
with party number, a distinction that is invisible to prior QCKA
analyses that assign $Q_{AB_i}$ and $Q_X$ as independent free
parameters rather than deriving them from a common physical channel.

\begin{figure}[!htbp]
    \centering
    \includegraphics[width=0.9\linewidth]{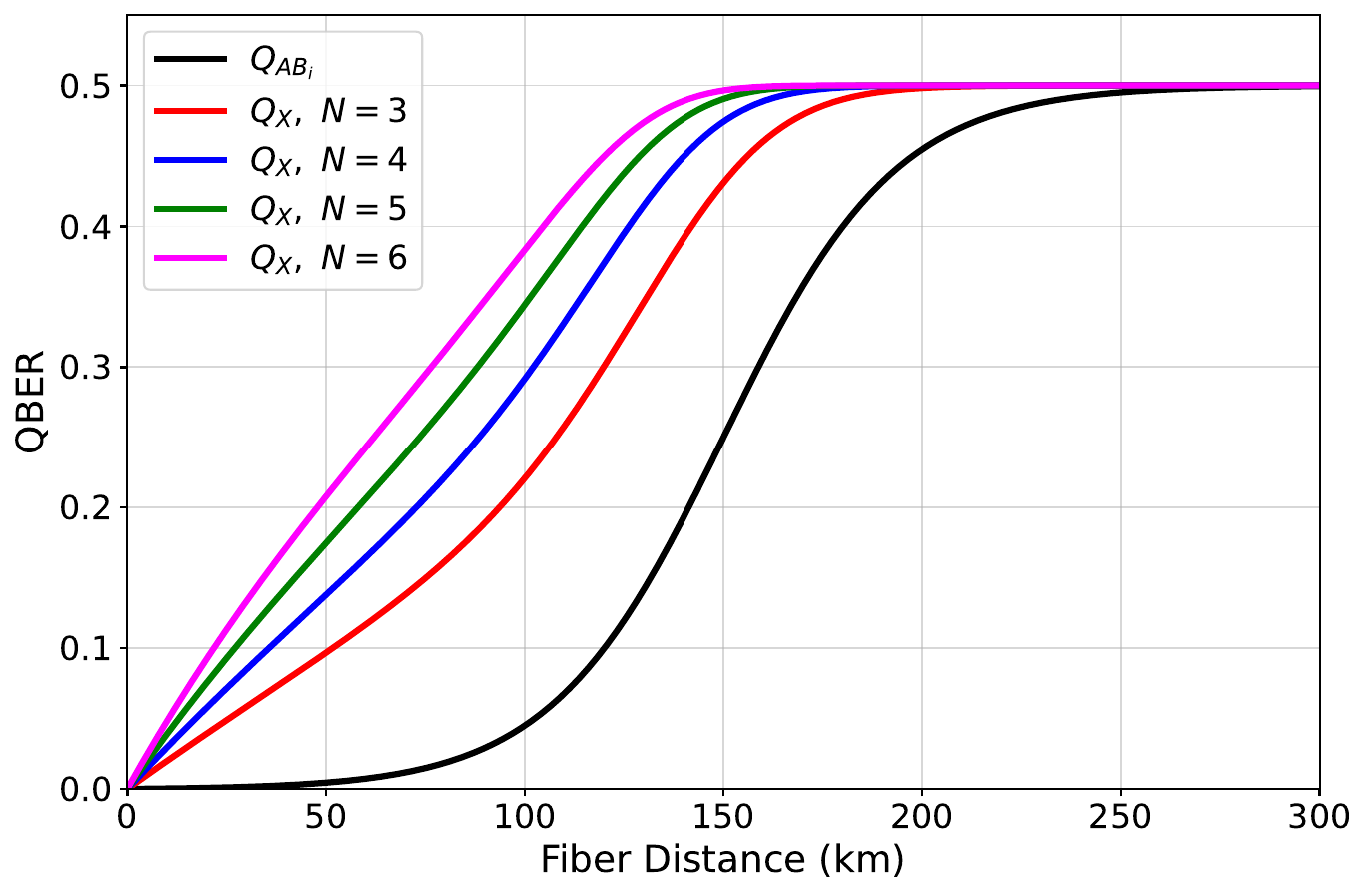}
    \caption{Variation of the $Z$-basis and $X$-basis QBERs with fiber distance for different numbers of parties $N$ under the erasure-Pauli channel.}
    \label{fig:qber_distance}
\end{figure}

Fig.~\ref{fig:comparison_optimization} compares the finite-key SKR obtained
with the fixed parameter-estimation probability $p=0.02$ against the
Brent-optimized probability $p^\star(N,d)$ obtained by solving
Eq.~\eqref{eq:optimization_problem}. For every $N$, optimization over
$p$ extends the maximum secure distance and increases the SKR
relative to the fixed-$p$ baseline. 
Furthermore, Table~\ref{tab:skr_summary} summarizes the SKR range for all
three scenarios: fixed $p$, optimized $p^\star$, and the asymptotic
case,  along with the resulting gain in SKR range by the
optimization.
Notably, the performance gain grows with $N$ because larger conferences operate closer to the noise threshold identified in
Fig.~\ref{fig:comparison_optimization}, they are more sensitive to the
statistical-fluctuation penalty $\xi(\varepsilon,n,m)$ in
Eq.~\eqref{eq:xi}, and therefore benefit disproportionately from a
test-sample allocation tuned to the actual channel conditions rather
than a one-size-fits-all convention. This result demonstrates that principled optimization of the parameter-estimation probability, rather than a fixed, conventionally chosen value, yields a practically significant and $N$-dependent improvement in both key rate and secure range for multipartite QCKA over a realistic channel. Since this gain is achieved purely through protocol-level tuning, with no changes to the fiber, source, or detector hardware, it represents a 'free' performance improvement directly relevant to real-world deployment in multipartite QKD networks.

\begin{figure}[!htbp]
    \centering
    \includegraphics[width=0.9\linewidth]{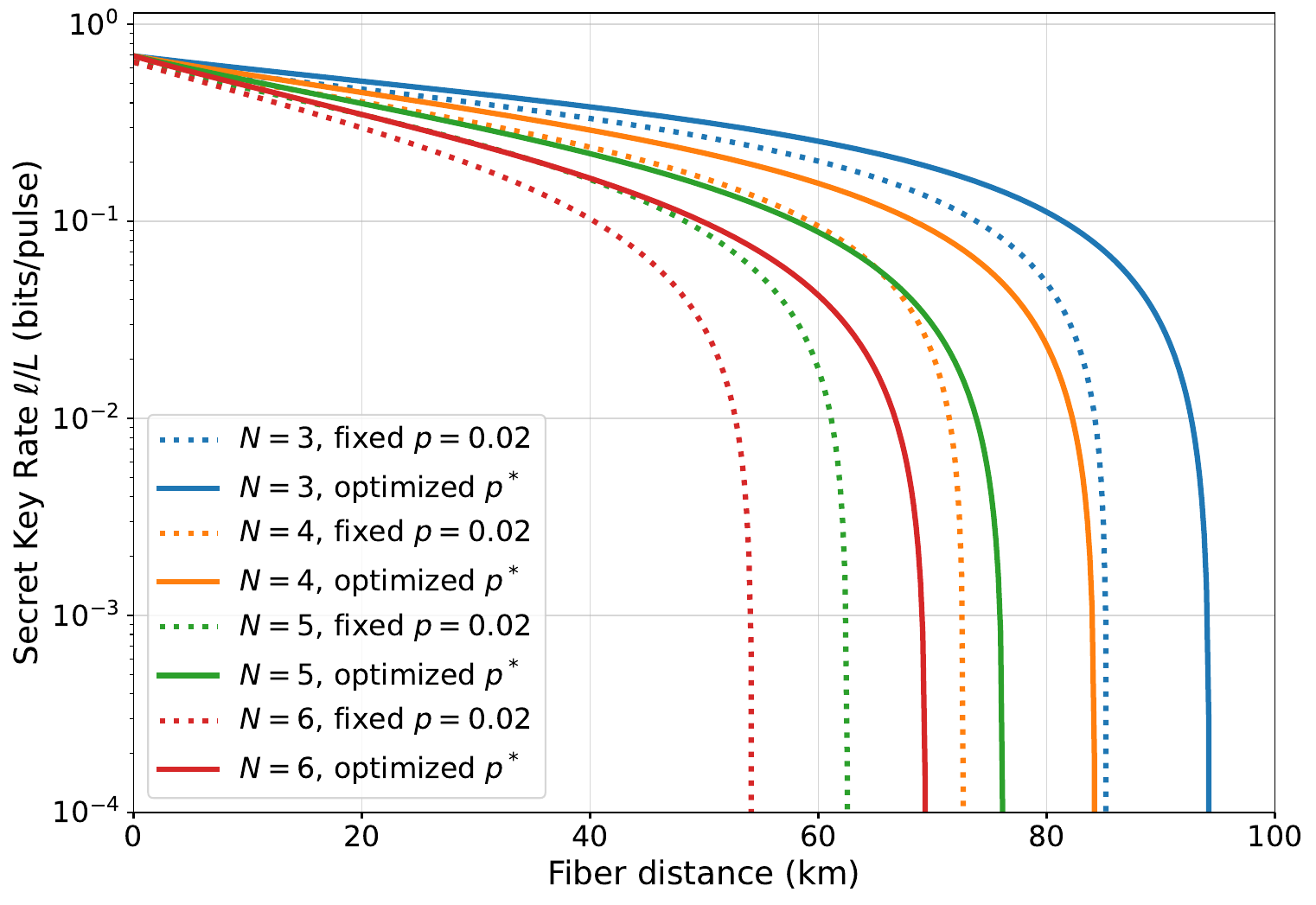}
    \caption{Comparison of the finite-key SKR as a function of
fiber distance for a fixed parameter-estimation probability
$p=0.02$ and the Brent-optimized probability $p^\star$ for different numbers of parties $N$.}
    \label{fig:comparison_optimization}
\end{figure}

Fig.~\ref{fig:comparison_asymptotic} compares the Brent-optimized
finite-key SKR against the asymptotic rate of Eq.~\eqref{eq:asymptotic_rate}
for each $N$. As expected, the asymptotic curve upper-bounds the
finite-key curve at every distance and $N$, since it neglects the
statistical fluctuation and leakage terms of
Eq.~\eqref{eq:theorem3}. The gap between the two- both in asymptotic SKR and in optimized SKR, quantifies the unavoidable
composable-security cost of finite-block operation. It reflects  the
overhead required to certify a fixed security parameter
$\varepsilon_{\mathrm{tot}}$ from a finite sample, on top of the
already-optimized protocol parameters. Reporting this gap alongside
the optimized finite-key result, rather than the asymptotic rate alone,
as is common in idealized treatments, gives a realistic
upper-performance benchmark against which future experimental
multipartite QCKA demonstrations over fiber can be assessed.

\begin{figure}[!htbp]
    \centering
    \includegraphics[width=0.9\linewidth]{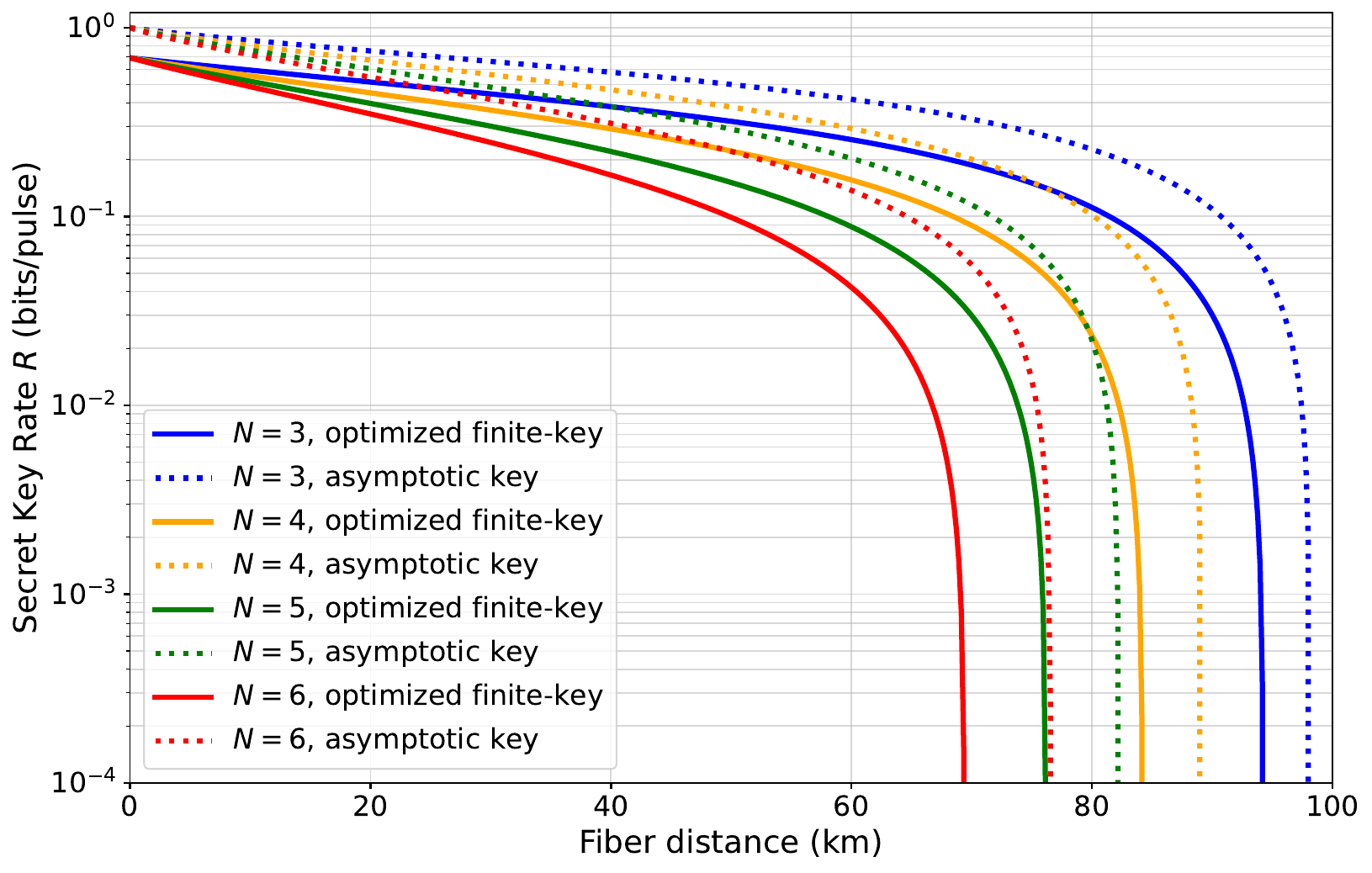}
    \caption{Comparison between the Brent-optimized finite-key SKR (solid curves) and the asymptotic SKR (dotted curves)
for different numbers of parties $N$ under the dark-count-enabled
erasure-Pauli channel.}
    \label{fig:comparison_asymptotic}
\end{figure}

\begin{table}[!htbp]
\centering
\caption{Summary of key output metrics (fixed-$p$, Brent-optimized, and asymptotic).}
\label{tab:skr_summary}
\begin{tabular}{ccccc}
\hline
\multirow{2}{*}{$N$} & \multicolumn{3}{c}{SKR range (km)} & SKR range gain \\
 & Fixed $p$ & Optimized $p^\star$ & Asymptotic & (opt. vs. fixed) \\
\hline
3 & 85.17 & 94.19 & 98.00 & 10.59\% \\
4 & 72.75 & 84.17 & 88.98 & 15.70\% \\
5 & 62.53 & 76.15 & 82.16 & 21.78\% \\
6 & 54.11 & 69.34 & 76.55 & 28.14\% \\
\hline
\end{tabular}
\end{table}

Fig.~\ref{fig:p_star} shows the optimized parameter-estimation probability
$p^\star$ as a function of fiber distance for every $N$.
At short distances, all configurations exhibit a similar
$p^\star$ of approximately $0.05$, indicating that the statistical
sampling requirement is relatively stable when the channel
errors are small. As the fiber distance increases, $p^\star$
increases progressively because the growing $Q_{AB_i}$ and
$Q_X$ require a larger parameter-estimation sample to maintain
the finite-key statistical bound $\xi(\epsilon,n,m)$ in Eq.~\ref{eq:xi}.
The increase occurs earlier for larger $N$, consistent with the
faster degradation of the global phase-error rate $Q_X$ observed
in Fig.~\ref{fig:qber_distance}. Thus, larger conference sizes require a greater fraction
of the available rounds to be devoted to parameter estimation
as the channel approaches its secure-distance limit. At the end of each curve, $p^\star$ approaches the numerical
search boundary $p_{\max}=0.5-10^{-4}$. This endpoint does not
represent a physically meaningful optimum; rather, it occurs
when the finite-key secret key length becomes zero for all
admissible values of $p$. The corresponding distance therefore
coincides with the collapse of the finite-key SKR observed in
Fig.~\ref{fig:comparison_optimization}. Hence, Fig.~\ref{fig:p_star} illustrates how the optimal allocation of
communication rounds adapts to increasing channel degradation
and how this allocation becomes constrained as the system
approaches its maximum secure transmission distance.

\begin{figure}[!htbp]
    \centering
    \includegraphics[width=0.9\linewidth]{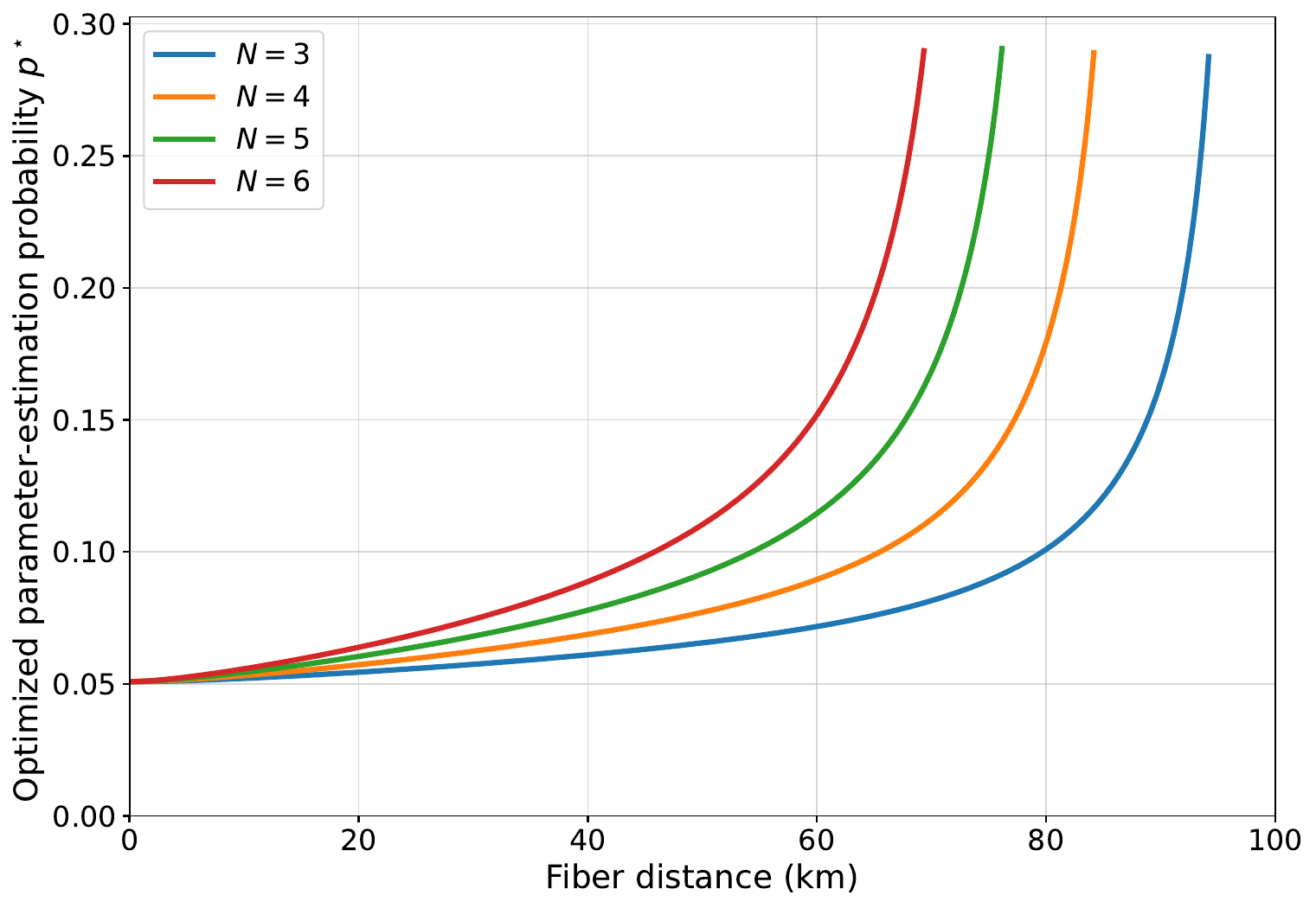}
    \caption{Brent-optimized parameter-estimation probability $p^{\star}$
as a function of fiber distance for different numbers of parties $N$. The curves increase with distance and approach the search bound $p_{\max}$ near the corresponding maximum secure distance.
}
\label{fig:p_star}
\end{figure}


Fig.~\ref{fig:qx_density} presents the variation of the phase-error rate $Q_X$ as a function of fiber distance and the PMD coefficient $D_{\mathrm{PMD}}$ for $N=3,4,5,$ and $6$ parties. The phase-error rate is evaluated using the GHZ-state coherence term after propagation through the dark-count-modified erasure--Pauli channel. For all configurations, $Q_X$ increases with both fiber distance and $D_{\mathrm{PMD}}$. The increase with distance results from the combined effect of fiber attenuation and polarization dephasing, while increasing $D_{\mathrm{PMD}}$ reduces the dephasing length $L_{\mathrm{DH}}$, thereby causing a faster loss of GHZ coherence. A clear dependence on the number of parties can also be observed. As $N$ increases, the coherence contribution is raised to the power $N-1$, making the phase-error rate more sensitive to channel imperfections. Consequently, the region of low $Q_X$ becomes progressively smaller as the number of parties increases from $N=3$ to $N=6$. In the high-distance and high-$D_{\mathrm{PMD}}$ region, $Q_X$ approaches its maximum value of $0.5$, indicating the loss of useful phase coherence. The red dashed contours mark the region where $Q_X$ approaches this saturation value. These results demonstrate that PMD-induced dephasing, together with fiber attenuation and increasing network size, significantly affects the phase-error performance of the multipartite QKD system.

\begin{figure}[!htbp]
    \centering
    \includegraphics[width=0.9\linewidth]{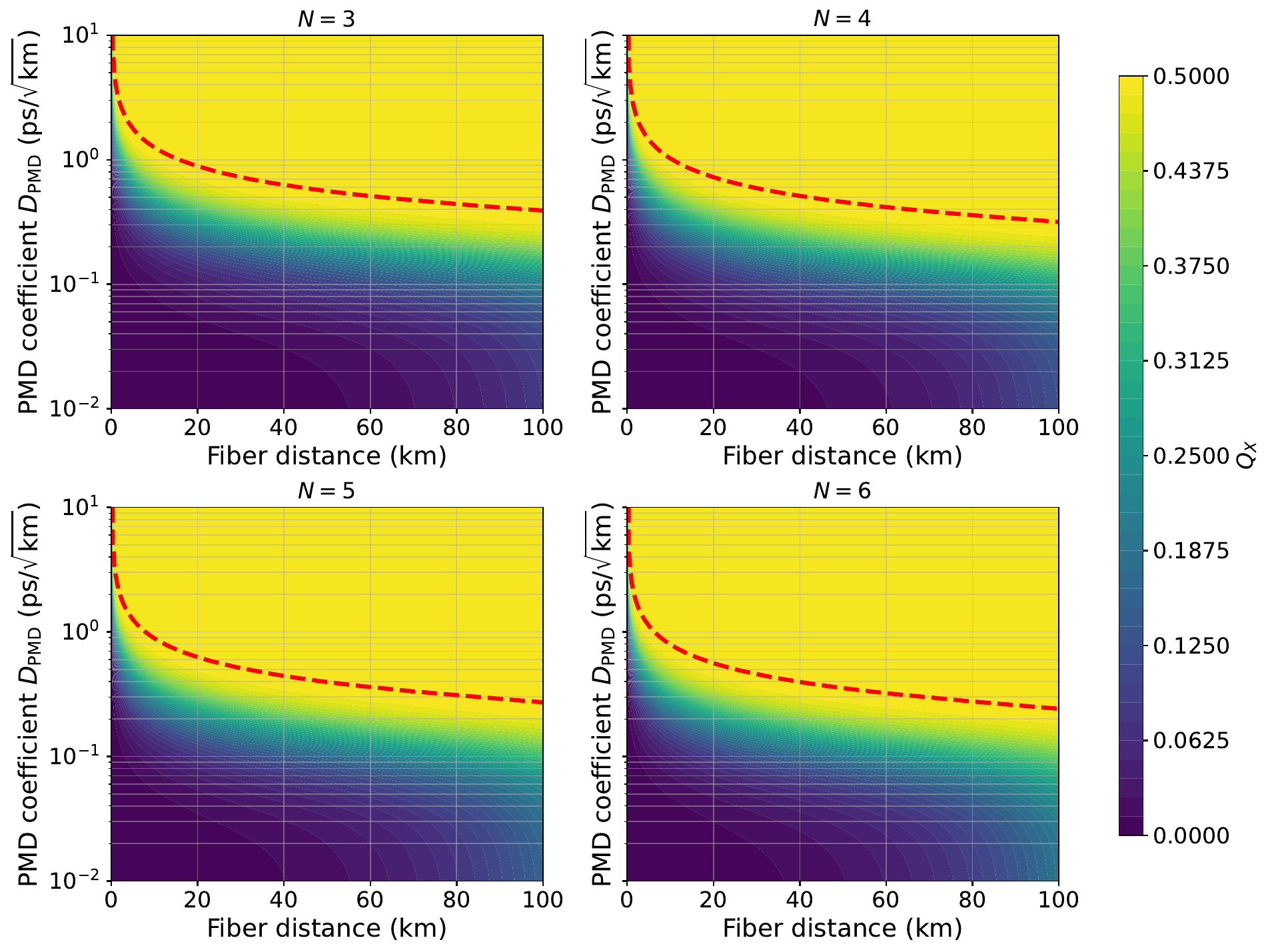}
    \caption{Phase-error rate $Q_X$ as a function of fiber distance $d$
and PMD coefficient $D_{\mathrm{PMD}}$ (log scale), shown separately
for $N=3,4,5,6$. Color encodes $Q_X$ from $0$ (dark purple, coherent)
to $0.5$ (yellow, fully decohered). The red dashed contour marks
$Q_X\approx0.5$, the practical decoherence boundary beyond which the
GHZ phase coherence is destroyed.
}
\label{fig:qx_density}
\end{figure}

Fig.~\ref{fig:skr_density} illustrates the dependence of the optimized finite-key secret key rate on the fiber distance and the PMD coefficient for different numbers of parties, $N=3,4,5,$ and $6$. The finite-key SKR is evaluated using the security expression in Eq.~\ref{eq:theorem3}, with the parameter-estimation probability $p$ optimized independently for each combination of fiber distance and PMD coefficient. As the fiber distance increases, the channel transmissivity decreases due to fiber attenuation, which increases the bit- and phase-error contributions and consequently reduces the achievable SKR. Similarly, an increase in $D_{\mathrm{PMD}}$ reduces the dephasing length $L_{\mathrm{DH}}$, thereby increasing the polarization-induced phase errors and further degrading the SKR. The results also show a clear dependence on the number of parties. As $N$ increases, the phase-error contribution becomes more pronounced because the GHZ coherence term depends on the number of transmitted receiver qubits. Consequently, the region supporting a non-zero secret key progressively decreases with increasing $N$. The $N=3$ configuration maintains a non-zero SKR over a comparatively broader range of fiber distance and PMD coefficient, whereas the secure region contracts for $N=4$, $5$, and $6$. The red dashed boundary indicates the transition toward a vanishing SKR and separates the parameter region in which secret-key generation is possible from the region where the finite-key rate becomes effectively zero. Overall, the figure demonstrates the combined impact of fiber attenuation, PMD-induced dephasing, and increasing network size on the finite-key performance of the proposed multipartite QKD system.

\begin{figure}[!htbp]
    \centering
    \includegraphics[width=0.9\linewidth]{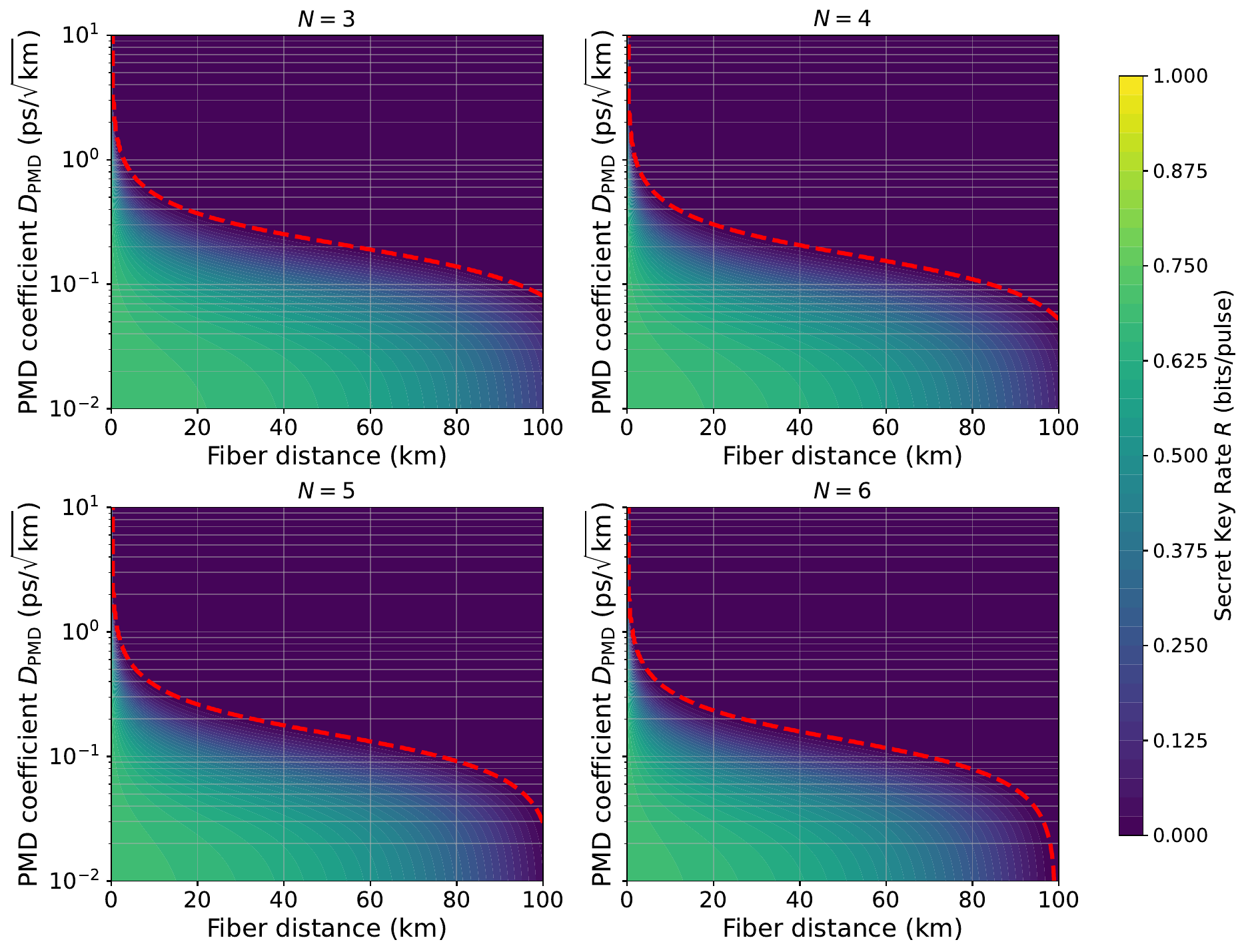}
    \caption{Two-dimensional density plots of the optimized finite-key SKR $R$ as a function of fiber distance and PMD coefficient $D_{\mathrm{PMD}}$ for different numbers of parties, $N=3,4,5,$ and $6$. The red dashed boundary indicates the transition
toward a vanishing SKR.
    }
    \label{fig:skr_density}
\end{figure}

\section{Conclusion}
\label{sec:conclusion}
In this work, we have analyzed the performance of GHZ-state-based
multipartite quantum communication over a dark-count-modified
erasure-Pauli channel. By explicitly propagating the GHZ density
matrix through this physical channel model, we derived closed-form
expressions for the local $Z$-basis bit-error rate and the global
$X$-basis phase-error rate and showed that the two are governed by
structurally distinct physical mechanisms. The former depends only on
fiber transmissivity and dark-count probability, while the latter
depends on the $(N-1)$-fold product of the surviving per-link
coherence factor and therefore degrades markedly faster with both
distance and party number. These channel-specific error rates were
integrated into a composable finite-key secret-key-rate framework,
and the parameter-estimation probability was optimized using Brent's
derivative-free method, which improved the achievable finite-key SKR
and maximum secure distance relative to a fixed-allocation baseline. Numerical results for varying a number of users, a
range of fiber distances, and PMD coefficients confirmed that
increasing fiber distance, PMD, and network size all reduce the
achievable SKR. Furthermore, comparison with the asymptotic SKR quantified the finite-key penalty arising from statistical fluctuations and security
correction terms. In this work we have assumed a homogeneous configuration with all Bobs equidistant from the source. Future extensions of this work could include extending the QCKA model to asymmetric configurations, where different Bobs experience different channel conditions. 


\bibliographystyle{IEEEtran}
\bibliography{ref}

\end{document}